\documentclass[a4paper,11pt]{article}
\usepackage[english]{babel}
\usepackage{mathrsfs}
\usepackage{makeidx}
\usepackage{amsfonts}
\usepackage[centertags]{amsmath}
\usepackage{amsthm}
\usepackage{enumitem}
\usepackage{latexsym,amsmath,amssymb}
\usepackage{graphicx,color}
\usepackage{comment}
\usepackage{graphics, setspace}
\usepackage{epstopdf}
\usepackage{tikz}
\usepackage{color}

\setenumerate[1]{label=(\roman*), ref=(\roman*)}
\usepackage[all]{xy}

\newtheorem{thm}{Theorem}[section]
\newtheorem{prop}[thm]{Proposition}

\newtheorem{lemma}[thm]{Lemma}

\newtheoremstyle{obs}% name
  {3pt}		%      Space above
  {3pt}		%      Space below
  {}		%      Body font
  {}		%      Indent amount (empty = no indent, \parindent = para indent)
  {\bfseries}%     Thm head font
  {.}		%      Punctuation after thm head
  {.5em}	%      Space after thm head: " " = normal interword space;
  {}		%      Thm head spec (can be left empty, meaning `normal')
\theoremstyle{obs}

\newcommand{\R}{\mathbb{R}}      %Numeros reales
\def\qed{\ifvmode\removelastskip\fi
{\unskip\nobreak\hfil\penalty50\hbox{}\nobreak\hfil \hbox{\vrule
height1.2ex width1.2ex}\parfillskip=0pt \finalhyphendemerits=0
\par \smallskip}}

\def\fpd#1#2{{\frac{\partial #1}{\partial #2}}}

\def\vectorfields#1{{\mathcal X}(#1)}
\def\cinfty#1{C^{\scriptscriptstyle\infty}(#1)}

\begin{document}

\title{The inverse problem of the calculus of variations and the stabilization of controlled Lagrangian systems}

\author{M.\ Farr\'e Puiggal\'i $\dagger$ and T.\ Mestdag $\ddagger$ \\[2mm]
{\small $\dagger$ Instituto de Ciencias Matem\'aticas (CSIC-UAM-UC3M-UCM)} \\ {\small  C/Nicol\'as
Cabrera 13-15, 28049 Madrid, Spain} \\ {\small marta.farre@icmat.es}
\\[2mm]
{\small $\ddagger$ Department of Mathematics and Computer Science, University of Antwerp,}\\
{\small Middelheimlaan 1, B--2200 Antwerpen, Belgium} \\
{\small and}
\\
{\small Department of Mathematics, Ghent University,}\\
{\small Krijgslaan 281, B--9000 Gent, Belgium} \\ {\small tom.mestdag@uantwerpen.be}
}

\maketitle

\begin{abstract}
We apply methods of the so-called `inverse problem of the calculus of variations' to the stabilization of an equilibrium of a class of two-dimensional controlled mechanical systems. The class is general enough to include, among others, the inverted pendulum on a cart and the inertia wheel pendulum. By making use of a condition that follows from Douglas' classification, we derive  feedback controls for which the control system is variational. We then use the energy of a suitable controlled Lagrangian to provide a stability criterion for the equilibrium.

\vspace{3mm}


\vspace{3mm}

\end{abstract}

%\tableofcontents

\section{Introduction}

With a view on achieving a desired goal dynamical systems are often modelled in such a way that a controlled quantity may influence its behavior. In this paper we will consider mechanical systems with a (possibly unstable) equilibrium. We will be interested in making structural modifications to this
system by adding extra controlled external forces or torques to it, in order to arrive
at a controlled system where the equilibrium has become stable. In a series of papers by Bloch {\em et
al.}\ (starting with the paper \cite{BLMfirst}) it was shown that, subject to a number of assumptions, some of
those controlled systems can be seen to be equivalent with the Euler-Lagrange equations of a new Lagrangian, the so-called
controlled Lagrangian. This controlled Lagrangian is a modification of the
original Lagrangian of the system by means of some control parameters. Sufficient conditions for this
situation to occur have been derived  in \cite{BLMfirst}  and the technique is often referred to as `the matching theorems'. Since its first appearance the method of controlled Lagrangians and the matching conditions have been successfully applied in many papers (see e.g.\ \cite{survey} for many references). The main advantage of the approach is that, once we know that the controlled system is Lagrangian, we may use energy methods and the available freedom in the choice of controls to analyze the stability of equilibria. A paper that focuses on the method for two-dimensional systems is e.g.\ \cite{Chang}.

In this paper we want to take a somewhat different approach to the matching theorems. The main idea is that we want to rephrase some aspects of the issue in terms of the so-called inverse problem of the calculus of variations.  Roughly speaking, the inverse problem of the calculus of variations concerns the question whether a given dynamical system can be
derived from a variational principle, i.e.\ whether it can be given, possibly in an equivalent form, by a set of Euler-Lagrange equations. If that is the case, we say that the dynamical system is `variational'. This challenging question has a
long history (see e.g.\ \cite{2008KP}): the necessary and sufficient conditions for this to happen are named after Helmholtz,
for example, and a distinguished paper on the solution of the problem for second-order systems with two
degrees of freedom is the one by Douglas \cite{Douglas}.

A recent paper that surveys both some aspects of the method of controlled Lagrangians and of the before-mentioned inverse problem is \cite{survey}.
The goal of the present paper is to give conditions for the stabilization of an unstable equilibrium for a concrete class of two-dimensional underactuated mechanical systems. We will come to our class of interest in two steps of specification. First we will assume that the Lagrangian of the original mechanical system with configuration variables $(x,y)$ is time-independent, that it has $x$ as cyclic variable, and that it is of the form
$$
\mathcal{L}(x,y,\dot x,\dot y)=\frac{1}{2}\left( a_{11}\dot{x}^2+2a_{12}(y)\dot{x}\dot{y}+a_{22}(y)\dot{y}^2 \right)-{\mathcal V}(y) \, ,
$$
where $a_{11}$ is a non-zero constant. We also assume that we may add controlled external forces to the system in such a way that the control subbundle is $\mbox{span}\left\{ dx \right\}$. The equations of motion are then of the type
\begin{eqnarray}
\frac{d}{dt}\left(\frac{\partial {\mathcal L}}{\partial \dot{x}}\right)&=&u \, , \nonumber\\
\frac{d}{dt}\left(\frac{\partial {\mathcal L}}{\partial \dot{y}}\right)-\frac{\partial {\mathcal L}}{\partial y}&=&0 \, . \label{controlsode}
\end{eqnarray}
 This class of systems is general enough to include, among others, the two main examples that have been discussed abundantly throughout the literature, namely the inverted pendulum on a cart and the inertia wheel pendulum (see Section~\ref{secex}). The  second-order ordinary differential equations  (\ref{controlsode}) can be written in normal form as
\begin{eqnarray*}
\ddot{x}&=&a^{12}\left( -\frac{\partial {\mathcal V}}{\partial y}-\frac{1}{2}\frac{\partial a_{22}}{\partial y}\dot{y}^2 \right)+a^{11}\left( -\frac{\partial a_{12}}{\partial y}\dot{y}^2+u \right) \, , \\
\ddot{y}&=&a^{22}\left( -\frac{\partial {\mathcal V}}{\partial y}-\frac{1}{2}\frac{\partial a_{22}}{\partial y}\dot{y}^2 \right)+a^{12}\left( -\frac{\partial a_{12}}{\partial y}\dot{y}^2+u \right) \, , \end{eqnarray*}
where $(a^{ij})$ is the inverse matrix of $(a_{ij})$. When we only consider controls of the form $u(y,\dot y)$, the above equations are of the type
\[
{\ddot x}=f^1(y, {\dot y}), \qquad {\ddot y}=f^2(y, {\dot y}).
\]
Our first goal is to understand when such a system is variational. For that purpose, we will rely (in Section~\ref{secvar}) on Douglas' classification \cite{Douglas} for two-dimensional systems, although we will use the geometric approach to the inverse problem that has been proposed in the papers \cite{CPST1999,94CSMBP,SCM1998,STP2002}  (see also Section~\ref{secDouglas}).
For most of the cases Douglas was able to decide whether or not the systems in it are variational. In our approach, the matching conditions are replaced with sufficient conditions for the system to lie in one of the variational cases of Douglas' classification.

If, in a second step, we only allow controls of the type $u(y,\dot y) =M(y)\dot{y}^2+N(y)$, the equations (\ref{controlsode}) may even be written in the form
\[
 \ddot{x}=T(y)\dot{y}^2+U(y),\qquad
\ddot{y}=R(y)\dot{y}^2+S(y).
\]
Our restriction in the second step is motivated by results in the literature. The condition that $a_{11}$ is constant is in fact (for two-dimensional systems) one of the so-called simplified matching conditions of \cite{BLMfirst}. Under these assumptions the authors derive for the system (\ref{controlsode}) a feedback control which may be written as
\begin{eqnarray} \label{cbloch}
u&=&\frac{1}{\sigma}\left( \frac{\partial a_{12}}{\partial y}-\frac{a_{12}}{A_{22}}\left( \frac{1}{2}\frac{\partial a_{22}}{\partial y}-\left( 1-\frac{1}{\sigma} \right)\frac{a_{12}}{a_{11}}\frac{\partial a_{12}}{\partial y} \right) \right)\dot{y}^2-\frac{1}{\sigma}\frac{a_{12}}{A_{22}}\frac{\partial {\mathcal V}}{\partial y} \, ,
\end{eqnarray}
 where $\sigma$ is a constant and $A_{22}=a_{22}-\frac{a_{12}^2}{a_{11}}\left(1-\frac{1}{\sigma}\right)$. This $u$ clearly fits into the class of controls that we wish to consider.

The strategy in the examples consists of pushing the controlled system into one of the cases of Douglas' classification that is known to be variational. In Section~\ref{secvar2} we will give a necessary and sufficient condition for a system of the above type to be variational.
Our approach is, in a sense, more general than the one of the matching conditions. In \cite{BLMfirst}, the matching conditions are a consequence of an a priori assumption on the relation between the original Lagrangian of the original system, and the controlled Lagrangian of the controlled system. In our approach, no such assumption needs to be imposed. Moreover, we will show in Section~\ref{secvar2} that if the system is variational, it admits a Lagrangian function of mechanical type, that is, a Lagrangian whose kinetic energy is related to a positive-definite metric. In that case, the energy function of this Lagrangian is always a first integral of the system.  We next show, in Section~\ref{stabres}, that under certain further conditions it can be used as a Lyapunov function. We conclude the section with a sufficient condition, written in terms of the system, that guarantees stability of the equilibrium.
	
In Section~\ref{secex} we discuss some examples. For the example of the inverted pendulum on a cart we give new feedback controls and we also recover the ones given in \cite{BLMfirst}.  For this class of controls, we provide a (slightly) wider class of Lagrangians.

The goal of Section~\ref{secas} is to achieve asymptotic stability by allowing dissipative forces into the picture. We first add in extra controls, to make the system equivalent to Euler-Lagrange equations with external dissipative forces. We then give sufficient conditions for asymptotic stability,  based on LaSalle's invariance principle. We illustrate this method by means of an example. In the final section we mention some directions for future work.

\section{The inverse problem} \label{secDouglas}

The inverse problem of the calculus of variations can be phrased in many subtly different versions, see e.g.\ \cite{2008KP,Saunders} for two review articles. The problem that we are concerned with poses the following question: Suppose given a system of second-order ordinary differential equations, given in normal form
\begin{equation} \label{sode}
\ddot{q}^i=f^i(t,q,\dot{q}), \qquad\qquad i=1,\ldots, n.
\end{equation}
 Is it possible to determine whether there exists a non-degenerate multiplier matrix $\left(g_{ij}(t,q,\dot{q})\right)$ such that the relation
\begin{equation} \label{IP}
g_{ij}(t,q,\dot{q})\left( \ddot{q}^j-f^j(t,q,\dot{q}) \right)=\frac{d}{dt}\left(\frac{\partial L}{\partial \dot{q}^i}\right)-\frac{\partial L}{\partial q^i}
\end{equation}
holds for some Lagrangian function $L(t,q,\dot{q})$? If such a multiplier matrix exists, we say that the system (\ref{sode}) is variational. In that case, the Lagrangian is related to the multiplier matrix in such a way that\
\[
\left(g_{ij}\right)=\left(\frac{\partial^2 L}{\partial \dot{q}^i \partial \dot{q}^j}\right).
 \]
 The non-degeneracy of the multiplier matrix ensures that the Lagrangian $L$ is regular. The so-called Helmholtz conditions are a set of necessary and sufficient conditions for a multiplier to exist. They constitute a mixed set of algebraic and PDE equations in the unknown functions $\left(g_{ij}(t,q,\dot{q})\right)$.

Douglas, in \cite{Douglas}, provided a classification of the problem for dimension $n=2$. For most of the subcases, he was able to conclude whether or not a Lagrangian exists. Douglas' analysis has led the authors of \cite{94CSMBP} to propose a generalization of the first broad classification of Douglas to arbitrary dimensions $n$, based on properties of the so-called Jacobi endomorphism $\Phi$ and the canonical covariant derivative $\nabla$. Both operators are essentially defined by the geometry of the `second-order ordinary differential equations field $\Gamma$' ({\sc sode} for short) that is generated by the system (\ref{sode}).  In the approach of \cite{CPST1999,94CSMBP,STP2002}, the system (\ref{sode}) is represented by the vector field
\[
\Gamma =\fpd{}{t}+ {\dot q}^i \fpd{}{q^i} + f^i(t,q,\dot q)\fpd{}{{\dot q}^i}
\]
on the first jet bundle $J^1\pi$ of a bundle $\pi:E=\R\times Q\to \R$. Here, $Q$ is the configuration manifold of the system. We will use the notation $\pi_1$ for the projection $J^1\pi \to E$. We will refer throughout the paper to sections of the pullback bundle $\pi_1^*(TE) \to J^1\pi$
as vector fields along $\pi_1$ and denote the set of such sections by $\vectorfields{\pi_1}$. For most of our purposes one may think of $Q$ as being $\R^n$, and of $\pi$ and $\pi_1$ as the projections $\R^{n+1} \to \R, (t,q) \mapsto t$ and $\R^{2n+1} \to \R^{n+1}, (t,q, {\dot q}) \mapsto (t,q)$, respectively. Vector fields along $\pi_1$ can then be represented by objects of the type $X^0(t,q,{\dot q}) \partial / \partial t + X^i(t,q,{\dot q}) \partial / \partial q^i$. One particular example of a vector field along $\pi_1$ is the so-called canonical vector field ${\mathbf T} = \partial / \partial t +{\dot q}^i \partial / \partial q^i$, but also vector fields on $Q$ and on $E$ can be thought of as being vector fields along $\pi_1$. In this way, one may see that the set $\{ {\mathbf T}, \partial / \partial q^i\}$ locally spans $\vectorfields{\pi_1}$.

It is well-known that the {\sc sode} $\Gamma$ defines a non-linear connection which ensures that every vector field $Z$ on $J^1\pi$  can be split in a horizontal and a vertical part (see e.g.\ \cite{94CSMBP,2008KP,willyreview}). This observation leads to the definition of two operators. The first, $\nabla: \vectorfields{\pi_1} \to \vectorfields{\pi_1}$, is a degree 0 derivation, which means that, for functions $F\in \cinfty{J^1\pi}$ and vector fields $X \in \vectorfields{\pi_1}$ along $\pi_1$, it satisfies
\[
\nabla(FX) = \Gamma(F)X + F \nabla X.
\]
For what follows, we only need its defining action on the basis $\{ {\mathbf T}, \partial / \partial q^i\}$:
\[
\nabla {\mathbf T} = 0, \qquad \nabla \fpd{}{q^j} = \Gamma^i_j(t,q,\dot q) \fpd{}{q^i}, \mbox{\quad with \quad} \Gamma^i_j =  -\frac{1}{2}\fpd{f^i}{{\dot q}^j}.
\]
The second operator $\Phi$ defines a (1,1)-tensor field along $\pi_1$, meaning that $\Phi(FX) = F\Phi(X)$. We may write it locally as
\[
\Phi = \Phi^i_j(t,q,\dot q) \fpd{}{q^i} \otimes (dq^j-{\dot q}^j dt), \mbox{\quad with \quad} \Phi^i_j = -\fpd{f^i}{q^j} - \Gamma^k_j\Gamma^i_k - \Gamma(\Gamma^i_j).
\]
The operation $\nabla$ can be further extended by duality to arbitrary tensor fields along $\pi_1$. In particular, $\nabla\Phi$ stands for the (1,1)-tensor field along $\pi_1$, given by
\[
(\nabla\Phi)(X) = \nabla(\Phi(X)) - \Phi(\nabla X).
\]
The coefficients of $\nabla\Phi = (\nabla\Phi)^i_j \partial/\partial q^i \otimes (dq^j-{\dot q}^j dt)$ are then
\begin{equation} \label{nablaPhi}
(\nabla\Phi)^i_j = \Gamma(\Phi^i_j) + \Gamma_m^i\Phi^m_j - \Gamma^m_j\Phi^i_m.
\end{equation}

The last operator we need is the vertical derivative  $D^v_X$. For each $X\in\vectorfields{\pi_1}$ it maps vector fields along $\pi_1$ to vector fields along $\pi_1$. It can be defined by requiring that it vanishes on  both ${\mathbf T}$ and the coordinate vector fields $\partial/\partial q^i$, and that it satisfies $D^v_X (FY)  = X^v(F)Y + F D^v_XY$ for all $F\in\cinfty{J^1\pi}$ and $Y \in \vectorfields{\pi_1}$.

The before mentioned Helmholtz conditions can be written in a form that makes use of the above geometric calculus. In e.g.\ \cite{SVCM} it is shown that a regular Lagrangian exists  for the system (\ref{sode}) if and only if there is a non-degenerate symmetric (0,2) tensor field $g$ along $\pi_1$  (i.e.\ a  multiplier) such that
\begin{equation} \label{Helmholtz}
g({\mathbf T},X) = 0,\qquad g(\Phi(X),Y) = g(X,\Phi(Y)), \qquad (D^v_Xg)(Y,Z) = (D^v_Yg)(X,Z), \qquad \nabla g=0 \, ,
\end{equation}
for arbitrary $X,Y,Z \in \vectorfields{\pi_1}$. We prefer to use in this paper this geometric approach to the Helmholtz conditions, over the more analytical style of Douglas' paper, for the reason that it can be conveniently applied (in the next section) to  a (non-coordinate) frame of eigenvectors of $\Phi$. More details on this calculus may be found in the review paper \cite{willyreview}.

The $\Phi$-condition represents an algebraic relation between the different components of the multiplier $g$. As such it forms the basis of the classification of the problem in several subcases.
For the rest of the paper, we will only consider two-dimensional systems (i.e.\ $n=2$). The first broad classification of \cite{94CSMBP} (and of \cite{Douglas}) is given by the following subcases:

- Case I: $\Phi$ is a multiple of the identity tensor $I$.

- Case II: $\nabla\Phi$ is a linear combination of $\Phi$ and $I$.

- Case III: $\nabla^2\Phi$ is a linear combination of $\nabla\Phi$, $\Phi$ and $I$.

- Case IV: $\nabla^2\Phi$, $\nabla\Phi$, $\Phi$ and $I$ are linearly independent.

\section{Discussion of Douglas' classification} \label{secvar}

As was mentioned in the Introduction, we are only interested in {\sc sode}s $\Gamma$ which exhibit very special symmetry properties. In this section we assume that there exists a coordinate change $(t,q^1,q^2) \mapsto (t,x=x(q^1,q^2),y=y(q^1,q^2))$ for which the second-order differential equations take the form
\begin{equation}\label{sodeofint}
{\ddot x}=f^1(y, {\dot y}), \qquad {\ddot y}=f^2(y, {\dot y}).
\end{equation}

\begin{lemma} The {\sc sode} $\Gamma$ takes the form (\ref{sodeofint}) if and only if $[\Gamma, \partial/\partial t]=0$ and if
there exists a vector field $E_1$ on $Q$ such that $\Phi(E_1)=\nabla E_1 =0$.
\end{lemma}
\begin{proof} The first condition says that the righthand sides of the second-order differential equations do not depend on $t$. If the {\sc sode} takes the special form (\ref{sodeofint}),  the vector field $E_1=\partial/\partial x$ satisfies the conditions. Conversely, if such a vector field $E_1 = X^i(q) \partial/\partial q^i$ on $Q$ exists, we may always straighten it out to become the vector field $\partial/\partial x$. In these coordinates, the condition $\nabla E_1 =0$ becomes $\Gamma^i_1 =0$, which means that the functions $f^i$ do not depend on ${\dot x}$. With that, the condition $\Phi(E_1)=0$ becomes $\Phi^i_1=0$, from which it follows that the functions $f^i$ do not depend on $x$ either. Hence, the system takes the form (\ref{sodeofint}).
\end{proof}

The specific form of the {\sc sode} (\ref{sodeofint}) narrows the number of cases in the Douglas classification to which it may belong. Since $\Phi(E_1)=0$, $\Phi$ has always  eigenvalue zero, with eigenvector $E_1=\partial/\partial x$. The other eigenvalue is given by $\Phi^2_2$. If non-zero, a corresponding eigenvector is
\[
E_2 = \fpd{}{y} + \nu \fpd{}{x}, \qquad \mbox{with\quad} \nu = \frac{\Phi^1_2}{\Phi^2_2}.
\]
The Douglas Case is therefore principally determined by the algebraic and geometric multiplicity of the zero eigenvalue.

Since also $\nabla\Phi(E_1)=\nabla^2\Phi(E_1)=0$, $\nabla^2\Phi,\nabla\Phi$ and $\Phi$ can never be pointwise linearly independent and therefore the system may never belong to Case IV. In coordinates where $E_1=\partial/\partial x$, the system will lie in Case I if and only if $\Phi^1_2 = \Phi^2_2 = 0$. It will belong to Case II when $\Phi^1_2$ and $\Phi^2_2$ are not both zero, but
\begin{equation} \label{caseIandII}
(\nabla\Phi)^1_2 \Phi^2_2 - (\nabla\Phi)^2_2 \Phi^1_2 =0.
\end{equation}

The system will belong to Case III whenever  $(\nabla\Phi)^1_2 \Phi^2_2 - (\nabla\Phi)^2_2 \Phi^1_2 \neq 0$. In that case, it is clear that the determinant of the commutator $[\Phi,\nabla\Phi]$ does always vanish, which is the defining property for the system to lie in subcase Case IIIb. Douglas has concluded in \cite{Douglas} that this case is never variational.

Case II has been further subdivided in Case IIa ($\Phi$ has distinct eigenvalues) and Case IIb (the eigenvalues of $\Phi$ coincide).
Both Cases IIa and IIb are further subdivided, according to a relation on a (1,2) tensor field along $\pi_1$, called the Haantjes tensor $H_\Phi(X,Y) = C^v_\Phi(\Phi(X),Y) - \Phi(C^v_\Phi(X,Y))$ of $\Phi$ in \cite{CPST1999,94CSMBP}. Case IIa1 and Case IIb1 correspond with the situation where $H_\Phi=0$. This tensor field vanishes when all the commutators $C^v_\Phi(X,Y) =  [D^v_X\Phi,\Phi](Y)$ vanish. For the {\sc sode} (\ref{sodeofint}) we get
\[
C^v_\Phi =  \left(\Phi^2_2 \fpd{\Phi^1_2}{{\dot y}} - \Phi^1_2 \fpd{\Phi^2_2}{{\dot y}} \right) dy \otimes dy \otimes \fpd{}{x}.
\]
From this expression, we may conclude that the last term in the Haantjes tensor always vanishes. The only non-vanishing term in the Haantjes tensor is then $H_{\Phi} (\partial/\partial y, \partial/\partial y) = \Phi_2^2 C^v_\Phi(\partial/\partial y,\partial/\partial y)$. The necessary and sufficient condition for the Haantjes tensor to vanish is therefore
\begin{equation}\label{Haan}
\Phi^2_2\left(\Phi^2_2 \fpd{\Phi^1_2}{{\dot y}} - \Phi^1_2 \fpd{\Phi^2_2}{{\dot y}}\right) = 0.
\end{equation}

If the system belongs to Case IIb (i.e.\ if $\Phi^2_2=0$) the above condition is trivially satisfied. Douglas \cite{Douglas} has one further subdivision of Case IIb1, depending on a further relation of the double eigenvalue of $\Phi$. In the special case when that double eigenvalue happens to be zero, Douglas' Case IIb1' is characterized by the vanishing of the expression
\[
\frac{\partial^2}{\partial {\dot x}^2} \left(\fpd{f^1}{\dot x} - \fpd{f^2}{{\dot y}}\right).
\]
This is clearly the case for the system (\ref{sodeofint}). We may therefore conclude that if the system (\ref{sodeofint}) belongs to Case IIb, it can only lie in Case IIb1'. For this case Douglas concluded that it is always variational.

Consider now the situation where the system (\ref{sodeofint}) belongs to Case IIa (i.e.\ $\Phi^2_2\neq 0$).  Since the Haantjes tensor has at most 1 non-vanishing component what is called Case IIa3 can never occur. The only possibilities are therefore Case IIa1 (with vanishing Haantjes tensor) and Case IIa2 (the 1 component of the Haantjes tensor does not vanish). Douglas concluded that Case IIa1 is always variational (the same is true in general dimension $n$, see \cite{CPST1999}). The necessary and sufficient condition for this to happen is
\[
\Phi^2_2 \fpd{\Phi^1_2}{{\dot y}} - \Phi^1_2 \fpd{\Phi^2_2}{{\dot y}}= 0.
\]
For a system in Case IIa2 to be variational, further requirements hold.

From all this we may conclude:
\begin{prop} If the {\sc sode} (\ref{sodeofint}) is variational then condition (\ref{caseIandII}) is satisfied. If the system satisfies the further assumption (\ref{Haan}), condition (\ref{caseIandII}) is both necessary and sufficient  for the system to be variational.
\end{prop}
\begin{proof} For systems of the type (\ref{sodeofint}) Case IV is excluded. If the system is variational, it can not belong to Case III, since Case IIIb is never variational. It must therefore lie in either Case I or II, which is characterized by the condition (\ref{caseIandII}). If (\ref{caseIandII}) and (\ref{Haan}) are both satisfied, the Haantjes tensor vanishes. If so, we must be either in Case IIa1 or Case IIb1', both of which are variational.
\end{proof}

Case I is characterized by the fact that both $\Phi^1_2=\Phi^2_2=0$. For Case IIb1, $\Phi_2^2 =(\nabla\Phi)^2_2=0$, but $\Phi_2^1 \neq 0$.
 For Case IIa1 $\Phi_2^2\neq 0$ and $(\nabla\Phi)^1_2 = \nu (\nabla\Phi)^2_2$.

\section{Conditions for variationality} \label{secvar2}
Our interest in systems of the type (\ref{sodeofint}) has been motivated in the Introduction by the fact that control systems of the type (\ref{controlsode}) with controls $u(y,\dot y)$ all fall in this category. In the second step
we limit the suitable controls to those of the quadratic type   $u(y,\dot y) =M(y)\dot{y}^2+N(y)$. As a result, the system  (\ref{controlsode}), when written in normal form becomes of the type
\begin{equation}
 \ddot{x}=T(y)\dot{y}^2+U(y), \\ \quad
\ddot{y}=R(y)\dot{y}^2+S(y). \label{SODEq}
\end{equation}
For later use, we give a few characterizations for its variationality below. In what follows we will denote a derivative  with respect to $y$ simply by a prime '.

\begin{prop} \label{pvar}
 The {\sc sode} (\ref{SODEq}) is variational if and only if
\begin{eqnarray} 0&=&
2 T (S')^2+S^2 \left(T R'-R T'\right)-2 R S' U'
+U' S''-S' U''   \nonumber\\ &\,&   \hspace{1cm} +S\left[S' T'+R^2 U'  -R' U'-T S''+R \left(-T S'+U''\right)\right].\label{rank1}
\end{eqnarray}

On the basis of the value of $\Phi^2_2$ we can further specify:
\begin{enumerate} \item When $\Phi^2_2 = 0$ the {\sc sode} (\ref{SODEq}) is always variational.
\item When $\Phi^2_2\neq 0$, the following statements are equivalent:
\begin{itemize} \item the {\sc sode} (\ref{SODEq}) is variational,
\item $(U-\nu S)'=0$,
\item  $\nu' = T-R\nu$.
\end{itemize}
\end{enumerate}
\end{prop}

\begin{proof}
 One easily verifies that for the system (\ref{SODEq}),
\begin{equation}\label{help1}
\Gamma^1_2 = -T {\dot y} , \quad \Gamma^2_2= -R {\dot y},\quad \Phi^1_2 = - U' + ST \quad\mbox{and}\quad \Phi^2_2 = -S' + RS,
\end{equation}
from which it follows that the condition (\ref{Haan}) is always satisfied. The necessary and sufficient condition for variationality is therefore condition (\ref{caseIandII}).

 Since now
\begin{equation} \label{help2}
(\nabla\Phi)^1_2 = {\dot y} (2S'T+ST'-U''-RU'), \quad (\nabla\Phi)^2_2 = {\dot y}(SR'+RS' -S'')
\end{equation}
the first statement in  the proposition follows.

When we take the value of $\Phi^2_2$ into account, we may further specify.

(i) We have already mentioned that the condition $\Phi^2_2=0$ is a sufficient condition for (\ref{SODEq}) to be variational since it implies $(\nabla\Phi)^2_2=0$. Therefore (\ref{caseIandII}) is satisfied.

(ii) In view of the coordinate expression (\ref{nablaPhi}) for $(\nabla\Phi)^i_j$ and the expressions (\ref{help1}) for $\Gamma^i_2$, we may write
\[
(\nabla\Phi)^1_2\Phi^2_2- (\nabla\Phi)^2_2\Phi^1_2 = [(\Phi^1_2)'\Phi^2_2- (\Phi^2_2)'\Phi^1_2 - \Phi^2_2( \Phi^2_2 T - \Phi^1_2 R )] {\dot y}.
\]
On the other hand, with $\nu=\Phi^1_2/\Phi^2_2$,
\begin{eqnarray*}
{\dot y}(U-\nu S)' &=& \frac{{\dot y}}{(\Phi^2_2)^2} \left[ U' (\Phi^2_2)^2 - \Big((\Phi^1_2)'\Phi^2_2- (\Phi^2_2)'\Phi^1_2 \Big) S - \Phi^1_2\Phi^2_2 S' \right] \\
&=&  -    \frac{ S}{(\Phi^2_2)^2} \Big((\nabla\Phi)^1_2\Phi^2_2- (\nabla\Phi)^2_2\Phi^1_2 \Big)     +   \frac{{\dot y}}{(\Phi^2_2)^2} \left[ - S \Phi^2_2( \Phi^2_2 T - \Phi^1_2 R )   +      U' (\Phi^2_2)^2  - \Phi^1_2\Phi^2_2 S'  \right] \\
&=&  -    \frac{ S}{(\Phi^2_2)^2} \Big((\nabla\Phi)^1_2\Phi^2_2- (\nabla\Phi)^2_2\Phi^1_2 \Big)
\end{eqnarray*}
because, in view of (\ref{help1}),
\[
- S ( \Phi^2_2 T - \Phi^1_2 R )   +      U' \Phi^2_2  - \Phi^1_2 S'  = -S(-S'+RS)T +S(-U'+ST)R +U'(-S'+RS) - (-U'+ST)S' =0.
\]

We may also write ${\dot y}\nu' = \Gamma(\nu) = (\Gamma(\Phi^1_2)\Phi^2_2-\Gamma(\Phi^2_2)\Phi^1_2)/(\Phi^2_2)^2$. With that
\[
\dot y \nu' = \frac{(\nabla\Phi)^1_2\Phi^2_2-(\nabla\Phi)^2_2\Phi^1_2}{(\Phi^2_2)^2} + \frac{\Phi^2_2 T - \Phi^1_2 R }{\Phi^2_2}\dot{y} = \frac{(\nabla\Phi)^1_2\Phi^2_2-(\nabla\Phi)^2_2\Phi^1_2}{(\Phi^2_2)^2} + (T-R\nu){\dot y}.
\]
\end{proof}

Remark that a sufficient condition for $\Phi^2_2=0$ is that $S(y)=0$.

The above proposition points to some strategies one may follow in the search for controls $u=M(y)\dot{y}^2+N(y)$ for which equations (\ref{controlsode}) are variational. For such a control law, equations (\ref{controlsode}) become of type (\ref{SODEq}) and the conditions given in Proposition~\ref{pvar} can be interpreted as a PDE in the unknowns $M(y)$ and $N(y)$. In a sense, one may interpret equation (\ref{rank1}) as a generalization (to the current setting) of the matching conditions of \cite{BLMfirst}. We may follow either one of the following paths:
\begin{itemize}
\item Find a control $u$ such that the corresponding {\sc sode} satisfies condition (\ref{rank1}).
\item Find a control $u$ for which $\Phi^2_2 \neq 0$, but the corresponding {\sc sode} satisfies $(U-\nu S)'=0$ (i.e.\ lies in Case IIa1).
\item Find a control $u$ such that the corresponding {\sc sode} satisfies $\Phi^2_2=0$ (i.e.\ lies in Case IIb1').
\item Find a control $u$ such that the corresponding {\sc sode} is such that $S=0$.
\end{itemize}
In this paper we will mainly concentrate on the first and  second strategies. The reason is that Case IIa1 has been shown to be variational in arbitrary dimensions \cite{CPST1999}, which leaves the door open to a possible generalization  of our results to higher dimensional systems. In the examples we will use an ansatz for $N(y)$ and solve the corresponding PDE for $M(y)$ (mainly because $N(y)$ appears with two derivatives in it and $M(\phi)$ with just one). In the next section we will also show that the last strategy is not the best one to follow, in view of the pursuit for stability.

The multipliers $g_{ij}$ of a variational system may in general depend on velocities ${\dot q}$. As a consequence, a Lagrangian of a variational system does not necessarily have to be of `mechanical' type, i.e.\ of the type 'quadratic kinetic energy -- potential'. We first prove that, if the system (\ref{SODEq}) is variational,  we may always find a Lagrangian $L$ of the form $L = g_{ij}{\dot q}^i{\dot q}^j-V(q)$, where  the matrix of multiplier matrix  $g_{ij}$ is independent of velocities, time-independent and positive definite.

\begin{prop} \label{posg}
For a variational {\sc sode} of type (\ref{SODEq}) with $\Phi^2_2\neq 0$ there exists a positive-definite matrix of multipliers $(g_{ij})$ which only depend on $y$ and for which $g_{11}$ is a constant.
\end{prop}

\begin{proof}
Under the assumptions in the statement the {\sc sode} belongs to Case IIa1. This means that the Jacobi endomorphism $\Phi$ has two distinct eigenvalues $0$ and $\Phi^2_2$, with eigenvectors $E_1$ and $E_2$, respectively. Remark that in the case under consideration both $E_1$ and $E_2$ may be thought of as vector fields on $Q$ (that is, when considered as vector fields along $\pi_1$, they do not depend on $t$ or on ${\dot q}$). One easily verifies that, after taking (\ref{caseIandII}) into account, we may write that
\[
\nabla E_1 = 0 , \qquad \nabla E_2 = \Gamma^1_2 \fpd{}{x} + \Gamma^2_2 \fpd{}{y} +\Gamma(\nu)\fpd{}{x} = \Gamma^2_2 E_2 = - R{\dot y} E_2,
\]
where we have invoked the third characterization of Proposition~\ref{pvar}.

We will denote the dual basis of 1-forms on $Q$ as $\{\theta^1,\theta^2\}$. From the above it follows that
  \[
  \nabla\theta^1 =0, \qquad \nabla\theta^2 = R{\dot y}\theta^2.
  \]
Since the system is supposed to be variational, we may  assume that solutions of the Helmholtz conditions (\ref{Helmholtz}) exist. We show now that among these solutions there is at least one that satisfies the specifics of the statement. From the $\Phi$-condition we may conclude that the multiplier is of the type $g=\rho_1 \theta^1 \otimes \theta^1 + \rho_2 \theta^2 \otimes \theta^2$. With this, the condition $\nabla g=0$ becomes
\[
\Gamma(\rho_1) =0, \qquad \Gamma(\rho_2) = -2 R {\dot y}\rho_2.
\]
We are not interested in the most general solution of these two PDEs in $\rho_i$. Any positive constant $\rho_1$ clearly satisfies the first equation, and we may even set it to be simply 1. We now show that the second equation has solutions $\rho_2(y)$ that only depend on $y$. Indeed, for such functions the equation becomes $\rho_2' = -2 R\rho_2$, which has (among other) the solutions $\rho_2(y) = A \exp(-2\int^y_1 R(\bar y) d{\bar y})$. Also the integration constant $A = \rho_2(1)$ can be chosen to be positive. With such functions $\rho_1=1$ and $\rho_2(y)$ the $D^v$-condition of the Helmholtz conditions is automatically satisfied. Clearly, $g= \theta^1 \otimes \theta^1 + \rho_2 \theta^2 \otimes \theta^2$ is then a positive-definite metric.
\end{proof}

In Proposition \ref{posg} we may replace positive-definiteness by negative-definiteness: In the proof we may choose $\rho_1$ and $A$ to be both negative.

\section{Lyapunov stability} \label{stabres}

In this section we  assume again that a mechanical system of the type (\ref{controlsode}) is given, with an arbitrary quadratic feedback control $u=M(y)\dot{y}^2+N(y)$. The relevant equations are then  of type (\ref{SODEq}).  For such systems $x$ is clearly a cyclic variable, and it generates a symmetry for the system. We may therefore reduce the two second-order differential equations in $(x,y)$ by that symmetry to a system of three first-order equations in $(y,v_y,v_x)$,  by cancelling out the variable $x$:
\begin{equation}\label{reducedsode}
{\dot y} = v_y, \quad \dot{v}_y = R(y)v_y^2 +S(y), \quad {\dot v}_x = T(y)v_y^2 +U(y).
\end{equation}
If we assume that $U(0) =S(0)=0$, the reduced system has an equilibrium at $(y=0, v_{x}=0, v_{y}=0)$ (or, equivalently, the original system (\ref{SODEq}) has a relative equilibrium $(x,y=0, \dot{x}=0,\dot{y}=0)$).

We wish to find sufficient conditions for that equilibrium to be stable. Note first that the Jacobian of the system (\ref{reducedsode}) in the equilibrium has a zero eigenvalue, and that, as a consequence, the equilibrium  can  never be linearly stable. Second, for systems of second-order differential equations, one may also consider a second, more geometric, linearization process, where the linearized equations are given by the matrix that corresponds to the Jacobi endomorphism $\Phi$ (see e.g.\ \cite{Sabau}). Since, for systems of the type (\ref{SODEq}), $\Phi$ has always eigenvalue zero, we can also not conclude that the equilibrium is Jacobi stable.

We are therefore left with trying to find a Lyapunov function for the system (\ref{SODEq}). For that reason, we now assume that we were able to find a feedback control $u=M(y)\dot{y}^2+N(y)$ for which the {\sc sode} (\ref{SODEq}) is variational, and for which $\Phi^2_2$ is not zero. Our method will rely on the use of the energy function of the variational system as a Lyapunov function (see e.g.\ \cite{Wiggins} for the definition of a Lyapunov function).

The multiplier we had found in the proof of Proposition~\ref{posg} may also be written in a coordinate basis as
\begin{eqnarray}
g&=&    dx \otimes dx - \nu (dx\otimes dy +dy\otimes dx) + \left(\nu^2 + \rho_2 \right) dy\otimes dy\nonumber\\
& =& g_{11}dx \otimes dx + g_{12} (dx\otimes dy +dy\otimes dx) + g_{22}dy\otimes dy.\label{gcoord}
\end{eqnarray}

Recall that the relation between possible Lagrangians and multipliers is such that the multiplier is the Hessian of the Lagrangian with respect to the velocities.  From this it follows that the Lagrangian that corresponds with the multiplier (\ref{gcoord}) must be of the type
\[
L =g_{11}{\dot x}^2 + 2g_{12}(y){\dot x}{\dot y} + g_{22}(y){\dot y}^2 +A_1(x,y){\dot x} + A_2(x,y){\dot y} - V(x,y).
\]
The Euler-Lagrange equations of $L$ provide the further conditions on the functions $A_i$ and $V$. We obtain
\begin{eqnarray*}
\fpd{V}{x} = - g_{11}U - g_{12}S,  && \fpd{A_1}{y} -\fpd{A_2}{x}=0, \\
\fpd{V}{y} = - g_{12}U - g_{22}S.  &&
\end{eqnarray*}
The equation which involves $A_i$ simply says that we may take any total time derivative $(\partial f/\partial q^i){\dot q}^i$ for the linear part $A_i{\dot q}^i$, for example simply $f=0$. The validity of the Helmholtz conditions (\ref{Helmholtz}) with the multiplier $g_{ij}$ ensures that a function $V(x,y)$ exists for the equations in the first column. This is clear form Proposition~\ref{pvar}, which shows that the integrability condition of this system of PDEs in $V$, namely $(g_{11}U + g_{12}S)'= (U-\nu S)'=0$, is guaranteed by the variationality of the system.

If we  assume as before that the system is such that  $S(0)=0$ and $U(0)=0$, then
\[
g_{11}U + g_{12}S = g_{11}(0)U(0) + g_{12}(0)S(0) =0.
\]
The  potentials which further satisfy $V(x,0)=0$ are then
\[
V(x,y) = - \int^y_0 (g_{12}U +g_{22}S)d\bar y =  \int^y_0(\nu U -\nu^2 S -  \rho_2 S)d\bar y = - \int^y_0 \rho_2 S d\bar y.
\]
We will denote this potential simply by $V(y)$. At $y=0$, it has the properties that
\[
\fpd{V}{y}(0)=0, \qquad \frac{\partial^2 V}{\partial y^2} (0) =- \rho'_2(0) S(0) - \rho_2(0) S'(0)=- \rho_2(0) S'(0).
\]
From the above we may conclude that, if we assume that $S'(0)<0$, then $y=0$ is a local minimum for $V$.

\begin{prop} \label{stability}
Suppose given a variational system (\ref{SODEq}) with $\Phi^2_2\not=0$, $U(0)=S(0)=0$ and $S'(0)<0$. Then  $(y=0, \dot{x}=0, \dot{y}=0)$ represents a stable relative equilibrium.
\end{prop}

\begin{proof}

Consider the energy function of the Lagrangian $L$ we had found above,
\begin{equation}
E_L(y,{\dot x},{\dot y})=\frac{1}{2}(g_{11}\dot{x}^2+2g_{12}(y)\dot{x}\dot{y}+g_{22}(y)\dot{y}^2)+V(y).   \label{energy}
\end{equation}
Since the Lagrangian is autonomous, this function is always a first integral of the system. It can now be used as a Lyapunov function. Indeed, since $V(0)=0$, we have $E_L(0,0,0)=0$. Since $y=0$ is always a stationary point for $V$, so will also be $(0,0,0)$ for $E_L$. Moreover since $g$ is positive-definite, and since $y=0$ is a minimum for $V$, we know that in a neighborhood of $(0,0,0)$, $E_L(y, \dot x,\dot y)>0$. We conclude therefore that $(y=0,{\dot x}=0,{\dot y}=0)$ is Lyapunov stable in the reduced space.
\end{proof}

Remark that, although the reasoning in the proof relies on the fact that we have chosen the multiplier matrix $(g_{ij})$ to be positive-definite, the condition $S'(0)<0$ does not. If we had chosen to work with a negative-definite multiplier, then $\rho_2$ would be negative, and with $S'(0)<0$ we would get that $(\partial^2 V/\partial y^2) (0)<0$, but then $E_L(y, \dot x,\dot y)<0$ in a neighborhood of $(0,0,0)$, which gives the same result.

\section{Examples} \label{secex}
\subsection{The inverted pendulum on a cart}

{\bf Definition of the system}. The system consists of a pendulum of length $l$ and a bob mass $m$. The pendulum is attached to the top of a cart of mass $M$. The configuration manifold of the system is $Q=S^1\times \mathbb{R}$ with coordinates $(x=s,y=\phi)$. The upright position of the pendulum corresponds with $\phi=0$. The Lagrangian is given by kinetic minus potential energy, that is,
$$
{\mathcal L}(s,\phi,\dot{s},\dot{\phi})=\frac{1}{2}(\gamma\dot{s}^2+2\beta\cos(\phi)\dot{s}\dot{\phi}+\alpha\dot{\phi}^2)+\delta\cos(\phi),
$$
where $\alpha=ml^2$, $\beta=ml$, $\gamma=M+m$ and $\delta=-mgl$ are constants related to the dimensions of the system, and $g$ denotes the standard acceleration due to gravity.

The control subbundle is $\mbox{span}\left\{ds\right\}$ and we have here
$$
a^{-1}(u ds)=u\left( -\frac{\beta\cos(\phi)}{\alpha\gamma-\beta^2\cos^2(\phi)}\frac{\partial}{\partial \phi}+\frac{\alpha}{\alpha\gamma-\beta^2\cos^2(\phi)}\frac{\partial}{\partial s} \right) \, ,
$$
where $a$ has components $a_{11}=\gamma, a_{12}=a_{21}=\beta\cos(\phi)$ and $a_{22}=\alpha$.
If we consider controls of the type $u(\phi,\dot\phi)=M(\phi){\dot\phi}^2+N(\phi)$, the controlled Euler-Lagrange equations (\ref{controlsode}), written in normal form, are
\begin{eqnarray*}
\ddot{s}&=&\displaystyle \frac{\beta \delta\sin(\phi)\cos(\phi)+\alpha\beta\sin(\phi)\dot{\phi}^2+\alpha u}{\alpha\gamma-\beta^2\cos^2(\phi)} \, , \\
\ddot{\phi}&=&\displaystyle \frac{-\gamma \delta\sin(\phi)-\beta^2\sin(\phi)\cos(\phi)\dot{\phi}^2-\beta\cos(\phi)u}{\alpha\gamma-\beta^2\cos^2(\phi)} \, .  \label{SODEc}
\end{eqnarray*}

\begin{center}
\begin{tikzpicture}
\draw (0,0) -- (8,0);
\draw (1,-1) -- (1,4);
\draw [->] (1,-0.5) -- (4,-0.5);
\node at (2.5,-0.7) {$s$};
\draw [dashed] (4,-1) -- (4,4);
\draw (2.7,0.3) rectangle (5.3,1.6);
\draw [thick] (4,1.6) -- (5,3.7);
\draw [fill] (5,3.7) circle [radius=0.2];
\draw [fill] (3.35,0.3) circle [radius=0.3];
\draw [fill] (4.65,0.3) circle [radius=0.3];
\draw [<-,domain=70:87] plot ({4+1.3*cos(\x)}, {1.6+1.3*sin(\x)});
\node at (4.3,3.3) {$\phi$};
\draw [->] (1.5,0.95) -- (2.5,0.95);
\node at (2,1.1) {$u$};
\draw [<->, dotted] (4.3,1.45714) -- (5.3,3.55714);
\node at (5.1,2.5) {$l$};
\end{tikzpicture}
\end{center}

{\bf A new stabilizing control.} We will give a new class of feedback controls which turn the upright position of the pendulum into a stable equilibrium, modulo the translational symmetry. For this purpose we look for solutions of the equation (\ref{rank1}). We will require that $\Phi^2_2\neq 0$, which means that we aim for a controlled {\sc sode} that lies in Case IIa1.

If we take $N(\phi)=d \cos(\phi)\sin(\phi)$, where $d$ is a constant, one may verify that the pair $(L,M)$ with
$$
M(\phi)=-\frac{d\left( 2\beta^2\delta-2\alpha\gamma \delta+\alpha\beta d+\beta(2\beta \delta+\alpha d)\cos(2\phi) \right)\sin(\phi)}{\delta\left( 2\gamma \delta+\beta d+\beta d \cos(2\phi) \right)} \,
$$
solves the PDE (\ref{rank1}). The controlled {\sc sode} is then given by
\begin{eqnarray*}
\ddot{s}&=&\left( \frac{\alpha\beta}{\alpha \gamma-\beta^2 \cos^2(\phi)}
-\frac{\alpha d \left(2 \beta^2 \delta-2 \alpha \gamma \delta+\alpha \beta d+\beta (2 \beta \delta+\alpha d) \cos(2\phi)\right)}{\delta (2 \gamma \delta+\beta d+\beta d \cos(2\phi))(\alpha \gamma-\beta^2 \cos^2(\phi))}\right)
\sin(\phi) \dot{\phi}^2 \\  && \hspace*{1cm}+  \frac{(\beta \delta+\alpha d) \cos(\phi) \sin(\phi)}{\alpha \gamma-\beta^2 \cos^2(\phi)}, \\  &=& T(\phi){\dot \phi}^2  + U(\phi) , \\
\ddot{\phi}&=& \left(\frac{\beta (\beta \delta+\alpha d)  (-2 \gamma \delta+\beta d+\beta
d \cos(2\phi)) \cos(\phi) \sin(\phi)}{\delta \left(\alpha \gamma-\beta^2 \cos(\phi)^2\right) (2 \gamma \delta+\beta
d+\beta d \cos(2\phi))}\right) \dot{\phi}^2  \\ && \hspace*{1cm}  -\frac{(2 \gamma \delta+\beta d+\beta d \cos(2\phi)) \sin(\phi)}{2 \left(\alpha \gamma-\beta^2
\cos^2(\phi)\right)} \\ &=&  R(\phi){\dot \phi}^2  + S(\phi) .
\end{eqnarray*}
We clearly have $U(0)=S(0)=0$. For the denominator in the expression for $\ddot\phi$, we have that
$$
\alpha\gamma-\beta^2\cos^2(\phi)=m^2l^2(1-\cos^2(\phi))+mMl^2>0 \mbox{ \,for all } \phi \, .
$$
If we fix some $\phi_{max}\in (-\pi/2,\pi/2)$, we may  choose $d$ in such a way that $d>\frac{2(M+m)g}{1+\cos(2 \phi_{max})}$. If so, we get that $2 \gamma\delta+\beta d+\beta d \cos(2\phi)>0$ in the range $(-\phi_{max},\phi_{max})$.

The  components of the Jacobi endomorphism for this {\sc sode} are given by
\begin{eqnarray*}
\Phi^2_2&=&\frac{\cos(\phi)\left(2 \gamma \delta+\beta d+\beta d \cos(2 \phi)\right) \left( -2\beta^2 \delta+2\alpha\gamma \delta-\alpha\beta d+\alpha\beta d\cos(2\phi) \right)}{ 2\delta \left(\alpha\gamma-\beta^2\cos^2(\phi) \right)^2 } \, ,\\
\Phi^1_2&=&-\frac{(\beta \delta+\alpha d)\cos^2(\phi)\left( -2\beta^2 \delta+2\alpha\gamma \delta-\alpha\beta d+\alpha\beta d\cos(2\phi) \right)}{\delta \left(\alpha\gamma-\beta^2\cos^2(\phi) \right)^2} \, .
\end{eqnarray*}

Since we also have $\alpha\beta d(\cos(2\phi)-1)-2\delta(\beta^2-\alpha\gamma)<0$ we get $\Phi^2_2\not =0$ for all $\phi\in (-\pi/2,\pi/2)$. This means that the {\sc sode} belongs to Case IIa1. By Proposition \ref{posg} we can find a positive-definite matrix of multipliers. Since $S'(0)=-\frac{\gamma \delta+\beta d}{\alpha\gamma-\beta^2}$ and $\alpha\gamma-\beta^2 = mMl^2>0$, we know from Proposition \ref{stability} that  the equilibrium $\phi=0, \dot{\phi}=0, \dot{s}=0$ will be stable in the reduced space when $d>\frac{-\gamma \delta}{\beta}=(m+M)g$.

If we fix the values of the parameters to be $M=2$, $m=1$, $l=1$ the control discussed above will stabilize the upright position of the pendulum for $d>3g$. If we also choose $\phi_{max}=\frac{\pi}{4}$ then we need to require $d>6g$ for the {\sc sode} to be defined. For these parameters and with $d=7g$ the {\sc sode} is
\begin{eqnarray*}
\ddot{s}&=&-\frac{6 \sin(\phi) \left(g \cos(\phi) (1+7 \cos(2 \phi))+(13+7 \cos(2 \phi)) \dot{\phi}^2\right)}{\left(-3+\cos(\phi)^2\right)
(1+7 \cos(2 \phi))} \, , \\
\ddot{\phi}&=&\frac{\sin(\phi) \left(g (1+7 \cos(2 \phi))^2+6 (33 \cos(\phi)+7 \cos(3 \phi)) \dot{\phi}^2\right)}{(-5+\cos(2 \phi))
(1+7 \cos(2 \phi))} \, .
\end{eqnarray*}

Below is a simulation of this example with {\sc matlab}, with initial conditions $\phi(0)=0.4,\, \dot{\phi}(0)=0.1,\, s(0)=0,\, \dot{s}(0)=-1.5$, and $g=9.81$. The position $s$ of the cart is not stabilized, since it represents a cyclic variable.

\includegraphics[height=9cm]{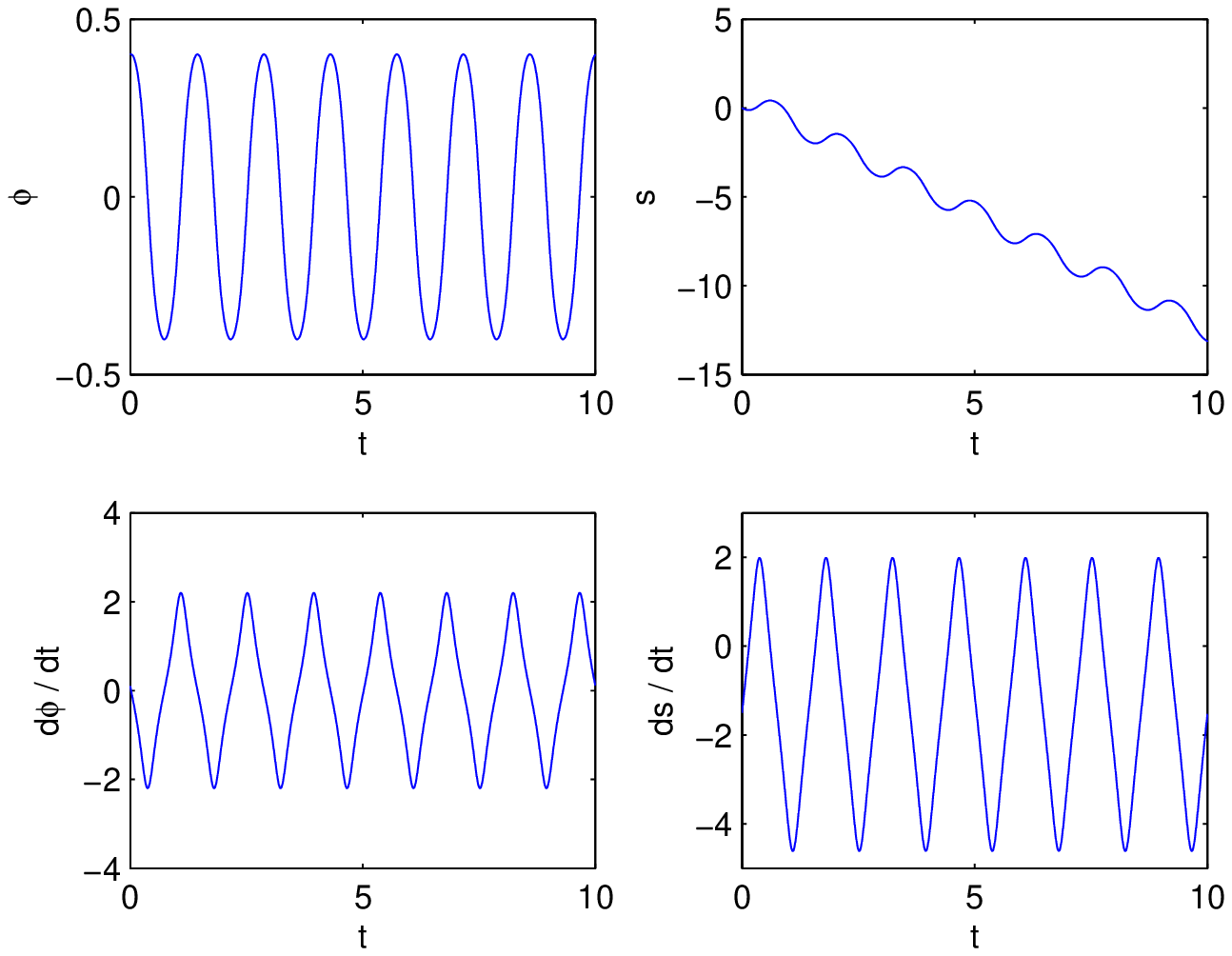}

{\bf The controls of \cite{BLMfirst}.} In this paragraph we recover the control given in \cite{BLMfirst} for the inverted pendulum on a cart and we give additional multipliers for the Lagrangian.
We consider again a control of the type $u(\phi,\dot{\phi})=M(\phi)\dot{\phi}^2+N(\phi)$, but now we take
\[
N(\phi)=\frac{\kappa\beta \delta \cos(\phi)\sin(\phi)}{\alpha-\frac{\beta^2}{\gamma}(1+\kappa)\cos^2(\phi)},
\] where $\kappa$ is a constant. With this $N(\phi)$,
$$
M(\phi)=\frac{\kappa\beta\alpha\sin(\phi)}{\alpha-\frac{\beta^2}{\gamma}(1+\kappa)\cos^2(\phi)} ,
$$
is a  solution of the PDE (\ref{rank1}). This control coincides with the one given in \cite{BLMfirst}. The controlled {\sc sode} is then
\begin{eqnarray*}
\ddot\phi &=&\frac{\gamma \delta\sin(\phi)}{-\alpha\gamma+\beta^2(1+\kappa)\cos^2(\phi)}+\frac{\beta^2(1+\kappa)\cos(\phi)\sin(\phi)}{-\alpha\gamma+\beta^2(1+\kappa)\cos^2(\phi)}\dot{\phi}^2 \, ,\\
\ddot s &=&-\frac{\beta \delta(1+\kappa)\cos(\phi)\sin(\phi)}{-\alpha\gamma+\beta^2(1+\kappa)\cos^2(\phi)}-\frac{\alpha\beta(1+\kappa)\sin(\phi)}{-\alpha\gamma+\beta^2(1+\kappa)\cos^2(\phi)}\dot{\phi}^2 \, .
\end{eqnarray*}

The value of $S'(0) = \gamma\delta/(\beta^2(1+\kappa)-\alpha\gamma)$ will be negative  when $\kappa>(\alpha\gamma-\beta^2)/\beta^2 =M/m$. For such values of $\kappa$ we will get a stable equilibrium.

Our criterion for stability does not involve the multipliers, nor the potential energy. If we compute the multipliers, we may compare them with the Hessian of the Lagrangian given in \cite{BLMfirst}. In the current setting
\[
\nu = -\frac{\beta (1+\kappa) \cos(\phi)}{\gamma} \, .
\]
The equation for $\rho_2$ was $\rho_2' = -2 R\rho_2$. One may easily verify that
\[
\rho_2 = A (\beta^2(\kappa+1) \cos^2(\phi)-\alpha\gamma )
\]
(with constant $A$) is a solution of it which is, however, not always positive. Since we are only interested in proving stability in a small region around the equilibrium, we may restrict our analysis to $\phi \in (-\phi_{max}, \phi_{max})$, with
\[
\sin^2(\phi_{max}) = \frac{\kappa - \frac{M}{m}}{1+\kappa}.
\]
(Under the current assumption on $\kappa$ the constant on the right hand is indeed positive.) The above region for $\phi$ coincides with the one that is also adopted in \cite{BLMfirst}. In this region, the denominator of the control never vanishes and, for every positive choice of $A$, the function $\rho_2$ remains positive and the multiplier matrix positive-definite. The corresponding multiplier is in fact
\begin{equation} \label{ourmult}
g_{11}=1, \qquad g_{12}= \frac{\beta (1+\kappa) \cos(\phi)}{\gamma}, \qquad g_{22}=\left( \frac{\beta^2(1+\kappa)^2}{\gamma^2}+2A\beta^2(1+\kappa) \right) \cos^2(\phi) -2A\alpha\gamma,
\end{equation}
from which one may derive, up to a constant,  the Lagrangian
$$
L=\frac{1}{2}\left( \dot{s}^2 +2g_{12}\dot{s}\dot{\phi} + g_{22}\dot{\phi}^2 \right)-2 A \gamma \delta\cos(\phi).
$$

The Lagrangian that has been proposed in \cite{BLMfirst} is
$$
L=\frac{1}{2}\left( \alpha\dot{\phi}^2 +2\beta\cos(\phi)\dot{s}\dot{\phi}+2\beta\cos(\phi)K\dot{\phi}^2 +\gamma\dot{s}^2 +2\gamma K \dot{s}\dot{\phi}+\gamma K^2 \dot{\phi}^2 \right)+\frac{\sigma}{2}\gamma K^2\dot{\phi}^2+\delta\cos(\phi)
$$
with $K=\kappa\frac{\beta}{\gamma}\cos(\phi)$, $\sigma=-1/{\kappa}$ and $\kappa$ a constant (satisfying $\kappa>\frac{\alpha\gamma-\beta^2}{\beta^2}$). Its multipliers are
\begin{equation}\label{blochmult}
g_{11}=\gamma, \qquad g_{12}=\beta(1+\kappa)\cos(\phi), \qquad
g_{22}=\frac{\beta^2\kappa}{\gamma}(1+\kappa)\cos^2(\phi)+\alpha.
\end{equation}
For better comparison, we may rescale both this multiplier and its Lagrangian with a constant factor $1/\gamma$ (to get also $g_{11}=1$). The multiplier matrix (\ref{blochmult}) then agrees with (\ref{ourmult}), if we set the integration constant $A$ to be $-1/(2\gamma^2)$. The negative choice for $A$ is not in disagreement with what we said before, since the multiplier matrix (\ref{blochmult}) of \cite{BLMfirst} is, surprisingly, non-definite in the region $(-\phi_{max}, \phi_{max})$.

\subsection{The inertia wheel pendulum}

{\bf Definition of the system.} The system consists of an inverted pendulum with an actuated wheel at the end. The configuration space is $S^1 \times S^1$. We will denote the coordinates of the system by $(x=\varphi,y=\theta)$, where $\varphi$ and $\theta$ are the angle of the wheel and the pendulum, respectively (see the figure below). The upright position of the pendulum corresponds to $\theta=0$. The Lagrangian is given by
$$
{\mathcal L}=\frac{1}{2}(b\dot{\varphi}^2+2b\dot{\theta}\dot{\varphi}+a\dot{\theta}^2)-m(1+\cos(\theta)) \, ,
$$
where $m$, $a$ and $b$ are positive constants with $a>b$.
These constants are defined from the physical parameters of the system as
$$
a=m_1 l_1^2+m_2l_2^2+I_1+I_2 \, , \quad b=I_2 \, , \quad \mbox{and} \quad m=m_1 l_1+m_2 l_2 \, ,
$$
where $m_1, I_1, m_2, I_2$ denote respectively the masses and moments of inertia of the pendulum and the wheel, and $l_1, l_2$ denote, respectively, the distances from the origin to the center of mass of the pendulum and the wheel, as shown in the picture below. See \cite{Spong}, for more details.

The controlled Euler-Lagrange equations are $a\ddot{\theta}+b\ddot{\varphi}=m\sin(\theta)$ and $b\ddot{\theta}+b\ddot{\varphi}=u$, which in normal form become
\begin{equation}
\ddot{\theta}= \frac{bm\sin(\theta)-bu}{b(a-b)}, \quad \ddot{\varphi}= \frac{-bm\sin(\theta)+au}{b(a-b)} \, .  \label{SODEiwp}
\end{equation}

\begin{center}
\begin{tikzpicture}
\draw (-1,0) -- (5,0);
\draw [dashed] (2,0) -- (2,4);
\draw [thick] (2,0) -- (4,3);
\draw [dashed] (4,3) -- (4.5,3.75);
\draw (4,3) circle [radius=0.5];
\draw [fill] (4,3) circle [radius=0.1];
\draw [fill] (3,1.5) circle [radius=0.1];
\draw [<-,domain=63:87] plot ({2+2*cos(\x)}, {0+2*sin(\x)});
\node at (2.5,2.3) {$\theta$};
\draw [<-,domain=0:57] plot ({4+0.7*cos(\x)}, {3+0.7*sin(\x)});
\node at (5,3.5) {$\varphi$};
\draw [<->,dotted] (2.2,-0.1334) -- (3.2,1.3666) ;
\node at (3,0.55) {$l_1$};
\draw [<->,dotted] (2.7,-0.4669) -- (4.7,2.5331) ;
\node at (4.2,1.1) {$l_2$};
%\node at (8,2) {$a=m_1 l_1^2+m_2l_2^2+I_1+I_2$};
%\node at (8,1) {$b=I_2$};
\draw [->,dotted] (1,1.5) -- (3,1.5) ;
\draw [->,dotted] (1,3) -- (4,3) ;
\node at (0.4,1.5) {$m_1, I_1$};
\node at (0.4,3) {$m_2,I_2$};
\end{tikzpicture}
\end{center}

{\bf Stabilizing control.} In view of the lack of quadratic terms in ${\dot \theta}$ in the above equations, we try to find a control  $u=N(\theta)$ (i.e.\ with $M(\theta)=0$) such that {\sc sode} (\ref{SODEiwp}) lies in Case IIa1. Equation (\ref{rank1}) is then
$$
\frac{4m\dot{\theta}(\sin(\theta)N'+\cos(\theta)N'')}{(a-b)b}=0.
$$
It admits a solution $N(\theta)=d_2+d_1 \sin(\theta)$, where $d_1$ and $d_2$ are integration constants. Since we want the state $(\theta=0,\dot{\theta}=0,\dot{\varphi}=0)$ to be an equilibrium we must take $N(\theta)=d_1 \sin(\theta)$. In that case
\begin{equation}
\Phi^2_2=\frac{2(d_1-m)\cos(\theta)}{a-b} \quad \mbox{ and } \quad  \Phi^1_2=\frac{2(ad_1-bm)\cos(\theta)}{b(b-a)} \, ,
\end{equation}
so we will have $\Phi^2_2\not=0$ around the equilibrium as long as we require $d_1\not=m$.

The controlled {\sc sode} (in Case IIa1) is then given by
\begin{eqnarray*}
\ddot\theta &=&\frac{(m-d_1)\sin(\theta)}{a-b}=S(\theta) \, , \\
\ddot\varphi&=&\frac{(-ad_1+bm)\sin(\theta)}{b(b-a)}=U(\theta) \, .
\end{eqnarray*}
By Proposition \ref{stability} it is enough to choose $d_1>m$ to get stability for the equilibrium $\theta=0$, $\dot{\theta}=0$, $\dot{\varphi}=0$.

We choose the parameters of the system to be $a=0.4846$, $b=0.0032$ and $m=37.98$ (as in a simulation of \cite{GMN2011}). If we set the constant in the control to be $d_1=60$ and take the initial conditions to be $\theta_0=0.0001$, $\varphi_0=0.1$, $\dot{\theta}_0=0.0001$ and $\dot{\varphi}_0=0.1$,  we get the {\sc matlab} simulation below.

\includegraphics[height=9cm]{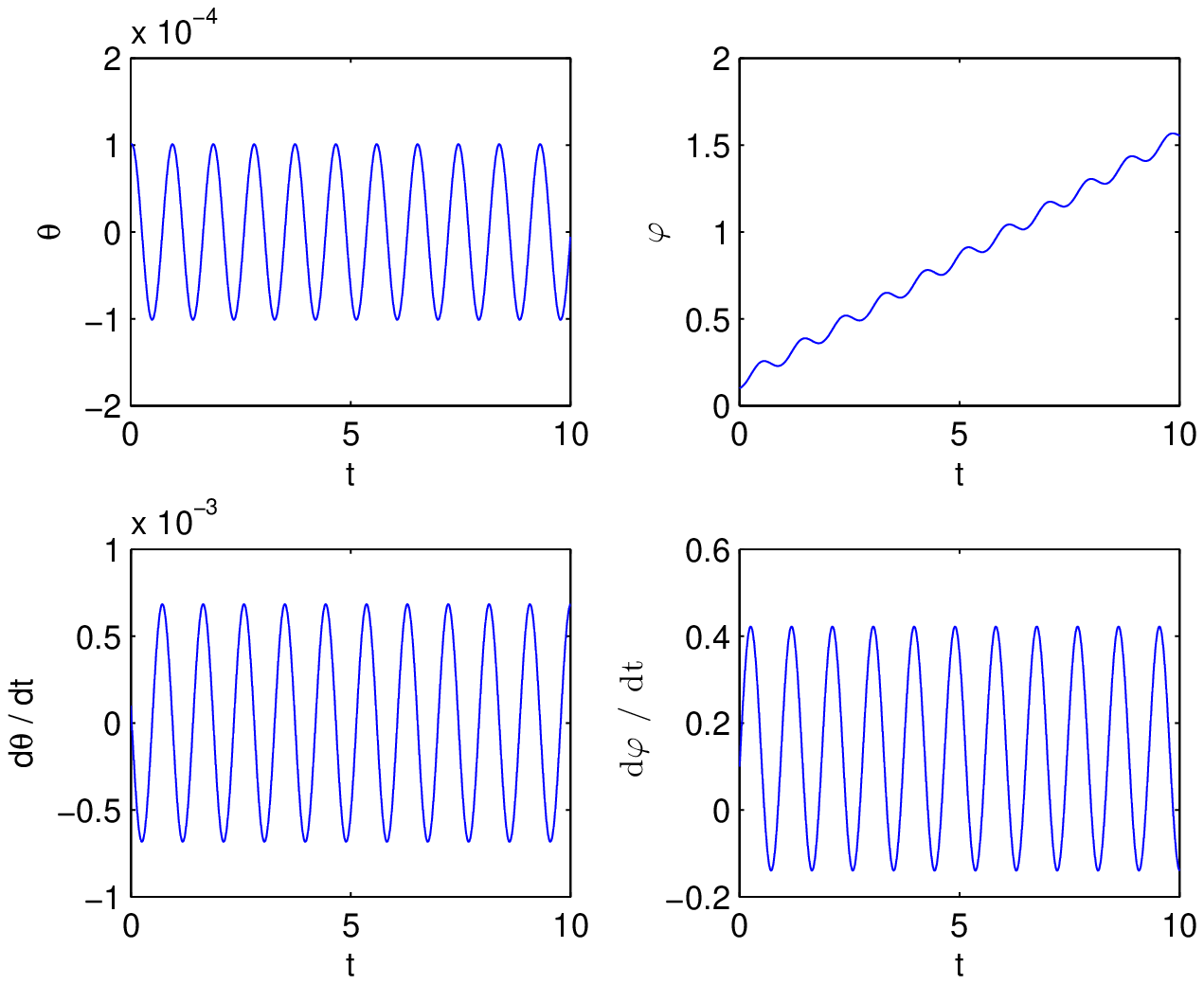}

In the Introduction we mentioned the feedback control (\ref{cbloch}) of \cite{BLMfirst}. Although the example is not explicitly treated in \cite{BLMfirst},  one may still calculate the corresponding control. Since the multipliers of the given Lagrangian are constant, it reduces to
$$
u=-\frac{1}{\sigma}\frac{a_{12}}{A_{22}}\frac{\partial {\mathcal V}}{\partial \theta}=\frac{-bm}{\sigma (a-b)+b}\sin(\theta).
$$
This coincides with our control $u=d_1 \sin(\theta)$ if we take $\sigma=\frac{b(m+d_1)}{d_1(b-a)}$. Our stability condition $d_1>m$ is then equivalent with $\frac{2b}{b-a} < \sigma < \frac{b}{b-a}$.

\section{Asymptotic stability} \label{secas}

In Section~\ref{stabres} we only gave a criterion for stability of Lyapunov type. Along the lines of e.g.\ \cite{BLMfirst}, we now  modify the control that gives Lyapunov stability in such a way that the system becomes dissipative, and the equilibrium asymptotically stable.

We use, as before, the notation $q^i$ for the variables $(x,y)$. Assume that we had found a control $u(y, \dot{y})=M(y)\dot{y}^2+N(y)$ for which  the system (\ref{controlsode}) is  variational, that is, assume that we know of multipliers ${g}_{ij}$ and a regular Lagrangian ${L}$ such that
\begin{equation}
{g}_{ij}\left( \ddot{q}^j-f^j \right)=\frac{d}{dt}\left( \frac{\partial {L}}{\partial \dot{q}^i} \right)-\frac{\partial {L}}{\partial q^i} \, . \label{var1}
\end{equation}
Now, we further add control forces to this system with the goal of modifying it into a set of Euler-Lagrange equations with external dissipative forces. More precisely, we put $u=M(y)\dot{y}^2+N(y)+u_2$ in (\ref{controlsode}). Then, in normal form, we are considering systems of the type
\begin{equation}\label{modsode}
\ddot{q}^j=f^j+a^{-1}(u_2dx)^j,
\end{equation}
where $a$ is the metric of the original Lagrangian and the second control $u_2$ is chosen in such a way that
\begin{equation} \label{diss1}
g_{ij}\left( \ddot{q}^j-f^j-a^{-1}(u_2dx)^j \right)=\frac{d}{dt}\left( \frac{\partial {L}}{\partial \dot{q}^i} \right)-\frac{\partial {L}}{\partial q^i}-\frac{\partial D}{\partial \dot{q}^i}.
\end{equation}
The term in $D$ is called a dissipative force if it has the effect that, along trajectories, the energy $E_L$ has the property  ${\dot E}_L <0$. In view of (\ref{var1}), condition (\ref{diss1}) will hold when
\begin{equation} \label{Dcond}
{g}_{ij}(a^{-1}u_2dx)^j = u_2 (g_{i1}a^{11}+g_{i2}a^{12}) =\frac{\partial D}{\partial \dot{q}^i}\, .
\end{equation}

We may think of the above as a PDE in $D$. If we introduce the simplified notation
$$
\Box={g}_{11}a^{11}+{g}_{12}a^{12}= a^{11} -\nu a^{12}\, , \quad \diamond= g_{12} a^{11}+g_{22}a^{12} = -\nu a^{11}+(\nu^2+\rho_2)a^{12} \, ,
$$
the integrability condition is $\frac{\partial u_2}{\partial \dot{y}}\Box= \frac{\partial u_2}{\partial \dot{x}}\diamond$  (mixed derivatives of $D$ coincide). The functions $u_2$ that satisfy this condition are of the type
$$
u_2=f(x,y)\left(\Box\dot{x} +\diamond\dot{y}\right)+g(x,y).
$$
With this,
$$
D=f(x,y)\left(\frac{\Box^2}{2}\dot{x}^2+\Box\diamond\dot{x}\dot{y}+\frac{\diamond^2}{2}\dot{y}^2\right)+g(x,y)\left(\Box\dot{x}+\diamond\dot{y} \right)+h(x,y)
$$
satisfies the condition (\ref{Dcond}).

We had already established in Section~\ref{stabres} that, when $\Phi^2_2\neq 0$, $S(0)=0,U(0)=0$ and $S'(0)<0$, there exists a positive-definite multiplier and a potential $V$ such that, in a neighborhood around $(x,y=0,\dot x=0,\dot y= 0)$, $E_L>0$. It is easy to see that, along trajectories of the system, $\dot{E}_L={\dot q}^i (\partial D/\partial {\dot q}^i)$. If we choose $g=h=0$  and $f$ to be a strictly negative function, then
\[
D=\frac{f}{2}\left( \Box\dot{x}+\diamond\dot{y} \right)^2, \qquad  u_2=f\left(\Box\dot{x} +\diamond\dot{y}\right)
\]
and $\dot{E}_L \leq 0$. It is also clear that the equilibrium does not change under the extra control law, since $u_2(x,y=0,{\dot x}=0,{\dot y}=0)=0$.
From LaSalle's invariance principle (see e.g. \cite{Wiggins}) it follows that if the only trajectory of (\ref{modsode}) contained in the set
$$
M=\left\{(x,y,\dot{x},\dot{y}) : \dot{E_L}=0 \right\}=\left\{(x,y,\dot{x},\dot{y}) : D=0 \right\}=\left\{(x,y,\dot{x},{\dot y}) : \Box\dot{x}+\diamond\dot{y} =0 \right\}
$$
is $(x,0,0,0)$, then the relative equilibrium is asymptotically stable.

\begin{prop} Assume that the system (\ref{SODEq}) is variational  with $\Phi^2_2\not=0$, and that $S(0) =U(0)=0$ and $S'(0)<0$. If there exists no solution $(x(t),y(t))$ of (\ref{modsode}), other than the equilibrium $(x,0)$, that satisfies
\begin{equation} \label{ascond}
(\Box T + \diamond R) {\dot y}^2 + \dot\Box {\dot x} + \dot\diamond \dot y+   \Box U + \diamond S = 0,
\end{equation}
then the relative equilibrium is asymptotically stable.
\end{prop}

\begin{proof}
Suppose there is a  solution that satisfies  $\Box\dot{x}+\diamond\dot{y}=0$. Then it has $u_2 = f\left(\Box\dot{x} +\diamond\dot{y}\right)= 0$ and thus also
\[
0 = \dot{\Box}\dot{x}+ \Box\ddot{x}+ \dot\diamond\dot{y} + \diamond\ddot y= (\Box T + \diamond R) {\dot y}^2 + \dot\Box {\dot x} + \dot\diamond \dot y+   \Box U + \diamond S.
 \]
\end{proof}
The condition (\ref{ascond}) will be useful in the example below.

{\bf Example 1: Asymptotic stabilization of the inertia wheel pendulum.} We will modify the control that we had found in Section \ref{secex}, in accordance with the considerations above.
We first  compute the multipliers of the new Lagrangian. Notice that
$$
\nu =- \frac{ad_1-bm}{bd_1-bm}.
$$
Since $R=0$, we may take the function $\rho_2$ of the multiplier to be any positive constant. The multiplier is then the constant matrix
\begin{eqnarray*}
g_{11}=1 , \qquad g_{12} = -\nu, \qquad
g_{22}= \nu^2 + \rho_2.
\end{eqnarray*}
The controlled {\sc sode}, with extra control $u_2$, is
\begin{eqnarray*}
\ddot\theta &=&\frac{(m-d_1)\sin(\theta)}{a-b}-\frac{bu_2}{ab-b^2} \, , \\
\ddot\varphi &=&\frac{(-ad_1+bm)\sin(\theta)}{b(b-a)}+\frac{au_2}{ab-b^2} \, ,
\end{eqnarray*}
where $u_2= f (\Box\dot{\varphi} +\diamond\dot{\theta} )$. Since also the matrix $(a_{ij})$ is constant, both $\Box$ and $\diamond$ are constants. With all this, condition (\ref{ascond}) takes a very simple form. If a solution $(\theta(t),\varphi(t))$, other than the equilibrium, exists in the set where $\Box\dot{\varphi} +\diamond\dot{\theta} =0$, then this solution also satisfies
\[
0=\Box U +\diamond S = \left(\Box \frac{-ad_1+bm}{b(b-a)}  +  \diamond\frac{m-d_1}{a-b} \right)\sin(\theta(t)) =\frac{(d_1-m)\rho_2}{(a-b)^2}\sin(\theta(t)).
\]
Since we had chosen $d_1\not=m$ and $\rho_2>0$, we get that the first factor never vanishes. The only possible solution with the above property is therefore given by $\sin(\theta(t))=0$, and thus $\theta(t)=0$. We may conclude that the equilibrium is asymptotically stable.

If we take $u_2=\frac{-0.1}{\nu^2}(\diamond\dot{\theta}+\Box\dot{\varphi})$, the same parameters and initial conditions as in Section~\ref{secex}, and $\rho_2=\nu^2$ then we get the following {\sc matlab} simulation:

\includegraphics[height=9cm]{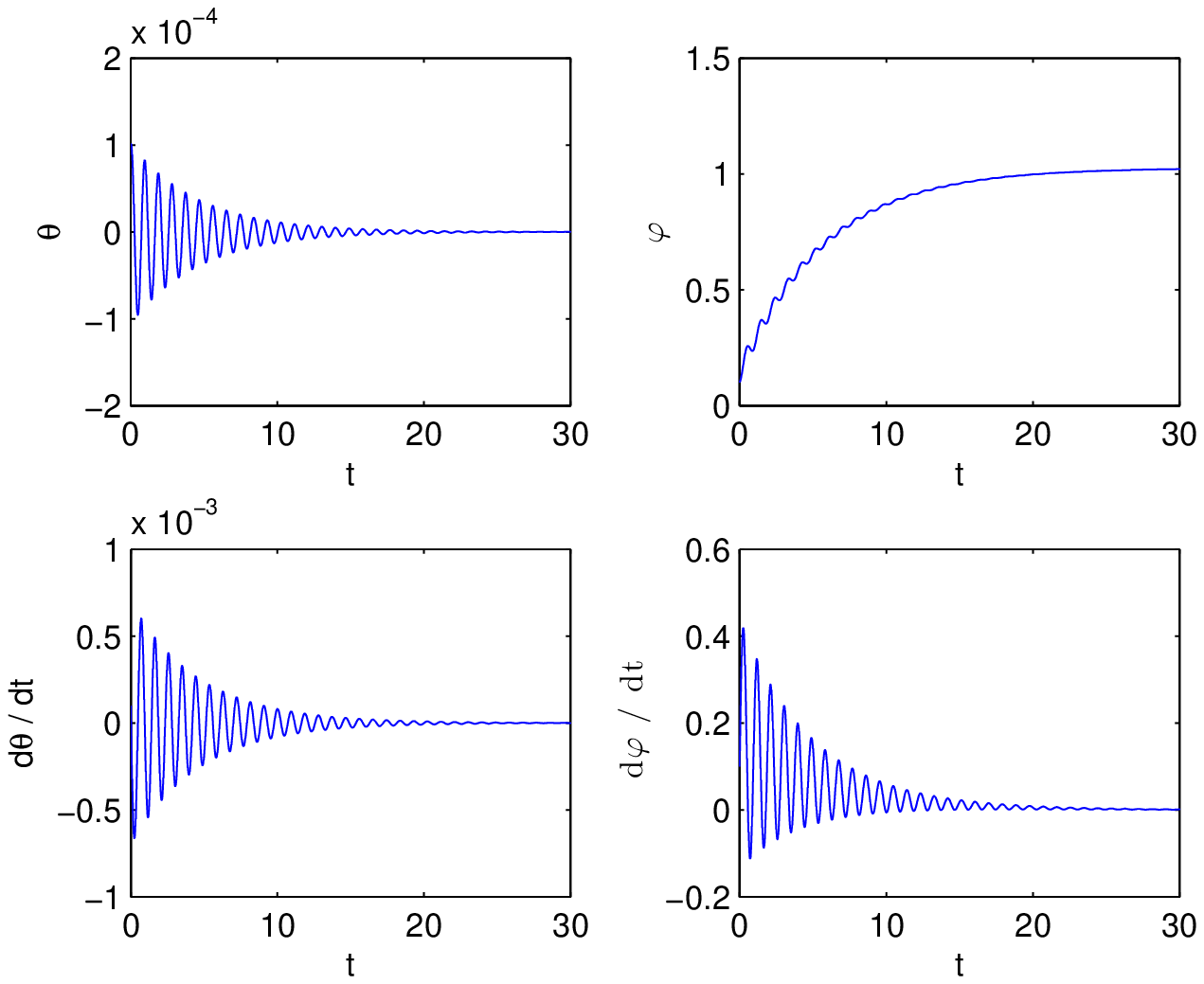}

{\bf Example 2: The inverted pendulum on a cart.} Consider again the new stabilizing control found in Section~\ref{secex}. In this case we get
$$
\nu=\frac{-2 (\beta \delta+\alpha d)\cos(\phi)}{(2 \gamma \delta+\beta d+\beta d \cos(2 \phi))} , \quad \rho_2=A\frac{(\beta^2-2\alpha\gamma+\beta^2\cos(2\phi))^{1-\frac{\alpha d}{\beta\delta}}}{(2\gamma\delta+\beta d+\beta d \cos(2\phi))^2} \, .
$$
If we choose $A=0.04$ and take the control $u_2=-0.03 s^2 (\Box \dot{s}+\diamond \dot{\phi})$, then with the same parameters and initial conditions as in Section~\ref{secex} we get the following {\sc matlab} simulation:

\includegraphics[height=9cm]{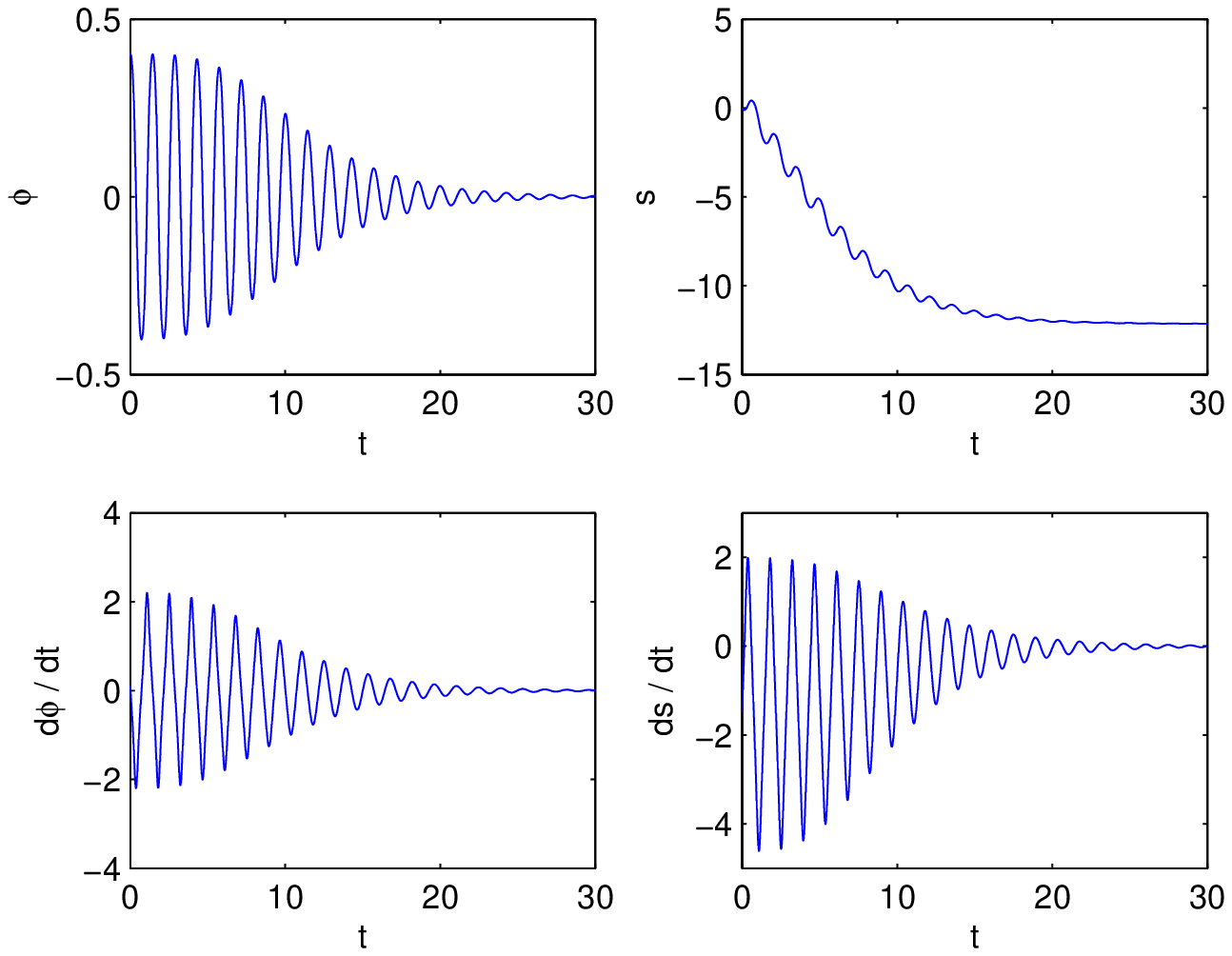}

\section{Conclusions and future directions}

In this paper we have presented a partially new technique by which one may stabilize a particular class of controlled Lagrangian systems with two degrees of freedom and one degree of actuation. The technique is based on finding a control law for which the system is pushed into one of Douglas' cases that are known to be always variational. For our class of systems we have given criteria to choose
\begin{itemize}
\item a feedback control which makes the system variational,
\item a feedback control which makes the system variational and stabilizes an unstable equilibrium,
\item a feedback control which makes the system Lagrangian with dissipative forces and asymptotically stabilizes an unstable equilibrium.
\end{itemize}

For a concrete problem, the $R$, $S$, $T$ and $U$ functions that appear in (\ref{rank1}) are known functions of $M$ and $N$. The first step in the proposed strategy is to solve equation (\ref{rank1}) for $M$ with a parametrized ansatz for $N$. The second step is to fix the parameters in such a way that $\Phi^2_2\not=0$ (in order for the controlled {\sc sode} to belong to Case IIa1) and $S'(0)<0$ (in order to obtain Lyapunov stability). Finally, if asymptotic stability is desired, then one may compute $\nu$ and $\rho_2$, and verify condition (\ref{ascond}).

We have used this strategy to find controls that stabilize the upright position of the inverted pendulum on a cart and of the inertia wheel pendulum.

As we mentioned in the Introduction, in \cite{CPST1999,94CSMBP} and follow-up papers, part of the Douglas classification is generalized to arbitrary dimensions, and it is shown that some of the subcases are always variational. One such variational case is precisely the Case IIa1 that we have exploited throughout this paper. We plan to study whether or not we can find controls that push {\sc sode}s of higher dimensional examples into one of the variational cases.

We also intend to study whether the condition of a cyclic variable can be removed. A paper in that direction is \cite{BCLM} where the method of controlled Lagrangians is extended to systems with a symmetry in the kinetic energy, but with a potential that breaks the symmetry. The price to pay is that, again, dissipative forces need to be brought into the picture. We wish to examine whether, by using the approach of this paper, one could  work without dissipative forces and find a suitable multiplier instead.

In the section on asymptotic stability we have used two steps. First we have assumed that we could add a control in such a way that the system is variational, second we have added an extra control to make the system dissipative, but with the same Lagrangian as the one of the first step. The first step is based on an analysis of the solution space of the Helmholtz conditions (\ref{Helmholtz}). However, also for dissipative systems, there exist Helmholtz-type conditions. That is to say, given a {\sc sode} (\ref{sode}) one could wonder (in any dimension $n$) under what conditions there exist a multiplier $g_{ij}$ and functions $L$ and $D$ such that
 $$
g_{ij}\left( \ddot{q}^j-f^j \right)=\frac{d}{dt}\left( \frac{\partial {L}}{\partial \dot{q}^i} \right)-\frac{\partial {L}}{\partial q^i}-\frac{\partial D}{\partial \dot{q}^i}.
$$
As it turns out, in \cite{CMS10,MSC2011}, it is shown that  necessary and sufficient conditions for this to occur can also be written entirely in terms of the multiplier $g_{ij}$, without having to make reference to the sought-for functions $L$ and $D$. One of the conditions is of algebraic type and of the form
 \[
\sum_{X,Y,Z} g(R(X, Y ), Z) = 0,
 \]
where $R$ stands for the curvature of the non-linear connection that can be associated to a {\sc sode}. A possible classification of such dissipative {\sc sode}s would be based on properties of this curvature and its derivatives (much like the Douglas classification is based on $\Phi$ and its corresponding Helmholtz condition).  In dimension $n=2$, however, the curvature condition is automatically satisfied and in \cite{Buca} it is even shown that every two-dimensional system of second-order differential equations is dissipative. It would be an interesting path to investigate whether, based on these ideas, one may find assumptions under which one may asymptotically stabilize a two-dimensional mechanical system.

{\bf Acknowledgments.}

MFP has been financially supported by MINECO (Spain) MTM 2013-42870-P, by the ICMAT Severo Ochoa project SEV-2011-0087 and by an FPU scholarship from MECD. She is grateful to the Department of Mathematics of Ghent University for its hospitality during the visit that made this work possible. Both authors thank David Mart\'{i}n de Diego for his helpful comments and careful reading of the manuscript and the referees for their constructive remarks.

%\bigskip

\bibliographystyle{plain}

\end{document}